\documentclass[11pt,reqno]{amsart}
\usepackage{amsmath,amssymb,mathrsfs,graphicx}
\usepackage[utf8x]{inputenc}
\usepackage[T1]{fontenc}
\usepackage[USenglish,english]{babel}
\usepackage{amsmath}
\usepackage{amsfonts}
\usepackage{latexsym}
\usepackage{cite}
\usepackage{float}
\usepackage{tikz}
\usepackage{pgfplots}
\usepackage{caption}
\usepackage{enumitem}
\usepackage{booktabs} 

\usepackage{subcaption}
\usepackage{mathtools}      
\mathtoolsset{showonlyrefs} 
\usepackage[export]{adjustbox}
\usepackage{amsthm}
\usepackage{amssymb,amscd}
\usepackage{mathtools}
\usepackage{xargs}
\usepackage{graphicx}
\usepackage{color}
\usepackage{enumitem}
\usepackage{multirow}
\usepackage{lmodern}
\usepackage{pdflscape}
\usepackage{rotating}
\usepackage[text={6.6in, 8.6in}, centering]{geometry}
 
\usepackage{parskip}
\usepackage{microtype}
\usepackage{diagbox}

\newtheorem*{thm*}{Theorem}

\newtheorem{prop}{Proposition}[section]

\newtheorem{cor}[prop]{Corollary}
\newtheorem{rmk}[prop]{Remark}

\usepackage{appendix}

\theoremstyle{definition}

\newtheorem{theorem}{Theorem}

\DeclareMathOperator{\sech}{sech}

 \title[]{A Betchov-Type Hydrodynamic Formulation of the Ivancevic Option-Pricing Equation}
\author[S. Kumar]{Sandeep Kumar}
\address[S. Kumar]{Department of Mathematics, CUNEF Universidad, Madrid, Spain}
\email{sandeep.kumar@cunef.edu}

\date{\today}	
\begin{document}

\begin{abstract}
We show that, under constant-coefficient assumptions, the Ivancevic option-pricing nonlinear Schr\"odinger equation admits a Betchov-type hydrodynamic formulation analogous to the one arising in the context of the vortex filament equation. Along with this structural correspondence, known Ivancevic solutions are interpreted in terms of density--velocity variables. As an illustration of the financial interpretation, we compare the normalized one-soliton density with the Black--Scholes risk-neutral density, obtaining an explicit relation between the Ivancevic coefficient ratio and standard market quantities. We further show that the probability current can be reconstructed from the temporal evolution of the density and demonstrate this reconstruction for an interacting two-soliton solution. Together, these results illustrate how the Betchov-type hydrodynamic formulation can be used to describe the evolution and transport of the Ivancevic density, while establishing a coefficient-explicit bridge between nonlinear wave models in mathematical finance and geometric fluid mechanics.
\end{abstract}


\maketitle
\section{Introduction}

Nonlinear Schrödinger equations (NLS) arise in many physical and applied settings, including nonlinear optics, water waves, plasma physics, Bose--Einstein condensates, vortex filament dynamics and more recently, in option pricing \cite{Carretero2008,Sulem2007,Contreras2010,Ivancevic2010}. One reason for their broad relevance is that the same amplitude--phase structure can encode different quantities in different contexts. For instance, in fluid dynamics, the nonlinear Schrödinger equation appears in the study of vortex filaments evolving under the vortex filament equation (VFE) \cite{Hasimoto1972}. First proposed by Da Rios \cite{Darios1906}, the VFE was further studied by Betchov, who described the intrinsic dynamics of the arc-length parameterized filament curve in terms of its curvature \(\kappa\) and torsion \(\tau\) \cite{Betchov1965}. Subsequently, with $s$ denoting arc-length parameter and $t$ time, Hasimoto showed that the complex wave function
\begin{equation}
	\psi(s,t)=
	\kappa(s,t)
	\exp\left(
	i\int^s \tau(\xi,t)\,d\xi
	\right),
	\label{eq:intro-hasimoto-transform}
\end{equation}
satisfies a focusing nonlinear Schrödinger equation, up to a time-dependent gauge transformation \cite{Hasimoto1972}. In this setting, the Hasimoto wave function \eqref{eq:intro-hasimoto-transform} may also be viewed through the Madelung variables
\[
\rho=|\psi|^2=\kappa^2,
\qquad
u=2\,(\arg\psi)_s=2\tau.
\]
With this normalization, Betchov's intrinsic equations can be written as a continuity equation and a momentum-type conservation law for \((\rho,u)\) \cite{Hasimoto1972}.
NLS-type equations have also appeared in econophysics and financial modeling
\cite{Stanley2000}. In mathematical finance, the classical Black--Scholes
model provides a foundational linear PDE for the price of a European option written on an underlying asset (e.g., a stock) \cite{BlackScholes1973}. Although this framework has been highly influential, it does not explicitly include nonlinear market
feedback, adaptive interactions or collective effects \cite{Jankova2018}.  Motivated by such features, Ivancevic proposed an adaptive-wave option-pricing model based on a
nonlinear Schrödinger equation \cite{Ivancevic2010}.
With $s$ denoting the underlying asset price and $t$ time, the option-price wave function $\psi(s,t)$ solves
\begin{equation}
	i\psi_t+\frac{\sigma}{2}\psi_{ss}
	+\beta|\psi|^2\psi=0,
	\label{eq:intro-ivancevic}
\end{equation} 
where $\sigma$ is interpreted as volatility which can either be fixed or stochastic and $\beta$ is an adaptive market-potential coefficient, which is related to the interest rate. Several analytic solutions of \eqref{eq:intro-ivancevic}, including bright solitons, dark waves, rogue-wave-type profiles and multi-soliton solutions, have been obtained by analytic and computer-algebra (symbolic-computation) techniques in the recent literature \cite{Chen2021,Delisle2026}.

The purpose of this paper is to develop a Betchov-type hydrodynamic
description of the Ivancevic option-pricing equation and to examine some of
its dynamical and financial consequences. Under the constant-coefficient
assumptions, when $\sigma$ and $\beta$ are taken to be real constants, we
consider the density--velocity formulation associated with
\eqref{eq:intro-ivancevic}. More precisely, using the Madelung transformation
\cite{Madelung1927}
\[
\psi(s,t)=\sqrt{\rho(s,t)}\,e^{i\theta(s,t)},
\]
we define
\begin{equation}
	\rho=|\psi|^2,
	\qquad
	u=\sigma\theta_s
	=
	\sigma\,\operatorname{Im}\left(\frac{\psi_s}{\psi}\right).
	\label{eq:intro-rho-u}
\end{equation}
These variables satisfy a Betchov-type conservation-law system
\begin{equation}
	\label{eq:convlawintro}
	\rho_t+J_s=0, \qquad J_t+K_s=0,
\end{equation}
where $J=\rho u$ is the probability current and $K$ is the associated
momentum flux, given explicitly in Theorem~\ref{thm:scalar-betchov}. In this
interpretation, $\rho=|\psi|^2$ is the option-price probability density over
asset price $s$ at time $t$; it is normalized when integrable (see
Remark~\ref{rmk:normalization} and Subsection~\ref{subsec:calibration}).
Additionally, $u$ is the phase-gradient velocity transporting this density
in the asset-price direction.

The passage from the nonlinear Schr\"odinger equation to the hydrodynamic
system \eqref{eq:convlawintro} is the classical Madelung
(quantum-hydrodynamic) rewriting of the NLS \cite{Madelung1927}. Here, we
make explicit its Betchov-type realization for the constant-coefficient
Ivancevic equation, retaining the dependence on the model coefficients.
More precisely, the volatility and adaptive-potential coefficients enter the
nonlinear-pressure and dispersive contributions to the momentum flux through
$\sigma\beta/2$ and $\sigma^2/4$, respectively. This yields a direct
correspondence with the Hasimoto--Betchov variables of vortex-filament
dynamics and provides a density--current interpretation of known
Ivancevic-type solutions, including bright and dark solitons, a two-component Manakov solution, and general $N$-soliton solutions.

Beyond this correspondence, we examine two consequences of the density-current interpretation in the financial setting. First, we compare the normalized one-soliton
density with the Black--Scholes risk-neutral density and obtain an explicit
relation between the Ivancevic coefficient ratio and standard market
quantities. Second, for time-resolved densities, the continuity equation
allows the probability current to be reconstructed from the temporal
evolution of the density, subject to a boundary condition. We demonstrate
this reconstruction for an interacting two-soliton solution, thereby
illustrating how the hydrodynamic formulation can be used to characterize
probability transport beyond a static density description.

Several transformations relating Schr\"odinger equations and Black--Scholes-type option-pricing models have been considered. The \emph{linear} Black--Scholes--Schr\"odinger correspondence is classical \cite{Contreras2010,Haven2004}; more recently, the \emph{linear} Black--Scholes equation has been reformulated in quantum-hydrodynamic variables via a Wick rotation \cite{Shushi2026}. In contrast, the present paper derives a Betchov-type hydrodynamic formulation of the nonlinear Ivancevic equation. To the best of our knowledge, this conservation-law structure has not previously been identified for the Ivancevic model.

The structure of the article is as follows. In Section~\ref{sec:Section2},
we derive the scalar Ivancevic--Betchov correspondence, discuss its extension
to fixed-polarization Manakov systems, and illustrate the formulation using explicit soliton solutions. In Section~\ref{sec:Section3}, we
discuss the financial interpretation of the density, current, and momentum
flux, including a static comparison with the Black--Scholes risk-neutral
density and the dynamical reconstruction of the probability current from
time-resolved density data. Finally, Section~\ref{sec:Section4} presents
outlook directions and open questions.	
\section{Betchov-type Formulation}
\label{sec:Section2}
\subsection{Hasimoto--Betchov background}	

Given arc-length parameter $s$ and $t$ time, let \(\mathbf X(s,t) \) be a three-dimensional curve, representing a vortex filament in inviscid fluid, with curvature $\kappa(s,t)$ and torsion $\tau(s,t)$.  Under the localized induction approximation, after a suitable normalization, the dynamics of the filament curve is described by \cite{Darios1906}
\begin{equation}
	\mathbf X_t= \mathbf X_s \wedge \mathbf X_{ss},
	\label{eq:VFE}
\end{equation}
where $\wedge$ is the cross-product. This equation is called the vortex filament equation or binormal flow. The Hasimoto transformation combines the two intrinsic quantities $\kappa$ and $\tau$ into the complex wave function
\begin{equation}
	\psi(s,t)
	=
	\kappa(s,t)
	\exp\left(
	i\int^s \tau(\xi,t)\,d\xi
	\right),
	\label{eq:hasimoto-transform}
\end{equation}
which relates \eqref{eq:VFE} with the nonlinear Schrödinger equation
\begin{equation}
	i\psi_t+\psi_{ss}+\frac12\psi(|\psi|^2+A(t))=0, \qquad A(t) \in \mathbb{R}.
	\label{eq:hasimoto-nls}
\end{equation}
Thanks to this relationship, the vortex filament equation has  been studied extensively, in both Euclidean and hyperbolic geometries; see, e.g., \cite{deLaHozKumarVega2020,deLaHozKumarVega2022,Kumar2022,Kumar2024} and references therein. Furthermore, by writing \eqref{eq:hasimoto-transform} in Madelung form,
\begin{equation}
	\psi(s,t)=\sqrt{\rho(s,t)}\,e^{i\theta(s,t)},
\end{equation}
one obtains
\begin{equation}
	\rho(s,t)=\kappa(s,t)^2=|\psi(s,t)|^2,
	\qquad
	u(s,t)=2\tau(s,t)=2\theta_s=2\, (\arg\psi)_s.		
\end{equation}
With this, the intrinsic equations obtained by Betchov in \cite{Betchov1965} can be written in conservation form as \cite{Hasimoto1972}
\begin{equation}
	\rho_t+(\rho u)_s=0,
	\label{eq:betchov-original-1}
\end{equation}
and
\begin{equation}
	(\rho u)_t+
	\partial_s
	\left[
	\rho u^2
	-\frac12\rho^2
	-\rho(\log\rho)_{ss}
	\right]=0.
	\label{eq:betchov-original-2}
\end{equation}
The first equation is a continuity equation for the curvature density $\rho$, while the second equation is a momentum-type conservation law. The term $-\frac12\rho^2$ is the nonlinear pressure term associated with the focusing NLS sign, while the term $-\rho(\log\rho)_{ss}$ is the dispersive or quantum-pressure contribution. Related geometric extensions to Manakov systems, involving two interacting
moving curves and an extended Da Rios--Betchov-type system, were considered
in \cite{Kostov2007manakov}.

As a motivating example, consider Hasimoto's bright soliton on a vortex filament:
\begin{equation}
	\kappa(s,t)
	=
	2\nu\,\sech\bigl(\nu(s-ct)\bigr),
	\qquad
	\tau(s,t)=\frac{c}{2}, \qquad \nu>0, c \in \mathbb{R}.
	\label{eq:hasimoto-bright-motivation}
\end{equation}
Then, we have
\begin{equation}
	\rho(s,t)
	=
	4\nu^2\sech^2\bigl(\nu(s-ct)\bigr),
	\qquad
	u(s,t)=c.
	\label{eq:hasimoto-rho-u-motivation}
\end{equation}
Thus, the classical filament solution corresponds to a localized density profile transported with constant velocity.  It is immediate to check that $\rho$ and $\rho u$ satisfy \eqref{eq:betchov-original-1}--\eqref{eq:betchov-original-2}.

\subsection{Ivancevic--Betchov correspondence}


In mathematical finance, an option is a derivative contract whose value depends on the value of an underlying asset (e.g., a stock). Given an asset price $s$ and time $t$, the classical Black--Scholes model describes the price function $V(s,t)$ of such a derivative by
\begin{equation}
	V_t+\frac12\sigma^2s^2V_{ss}+rsV_s-rV=0,
	\label{eq:black-scholes}
\end{equation}
where $\sigma$ is the volatility of the underlying asset and $r$ is the risk-free interest rate \cite{BlackScholes1973}.

The equation \eqref{eq:black-scholes} is linear and provides a foundational framework for pricing European options, i.e., options that can be exercised only at the maturity time \cite{BlackScholes1973}. However, it does not explicitly include nonlinear feedback effects, adaptive market interactions, or collective market responses \cite{Jankova2018}. To incorporate such effects, Ivancevic proposed an adaptive-wave option-pricing model in which the option-price wave function $\psi(s,t)$ satisfies the nonlinear Schrödinger equation
\begin{equation}
	\label{eq:ivancevic-nls}
	i\psi_t+\frac{\sigma}{2}\psi_{ss}
	+\beta|\psi|^2\psi=0.
\end{equation}
Here, $\sigma$ is interpreted as volatility, and $\beta$ is an adaptive market-potential coefficient and it can be interpreted as the interest rate. In this formulation, the squared modulus $\rho(s,t)=|\psi(s,t)|^2$ can be understood as a probability density over possible values of the underlying asset price at time $t$. 

Thus, the bright soliton example \eqref{eq:hasimoto-bright-motivation} not only gives the classical vortex-filament template for the density--velocity formulation, but also motivates us to look for a density--velocity representation for the Ivancevic option-pricing model. More precisely, we obtain the following results in this direction. 
\begin{theorem}[Scalar Ivancevic--Betchov correspondence]
	\label{thm:scalar-betchov}
	Let $\psi$ be a non-vanishing smooth solution of the Ivancevic nonlinear
	Schrödinger equation \eqref{eq:ivancevic-nls} on a domain $\Omega\subset\mathbb R_s\times\mathbb R_t$, with real constants
	$\sigma\neq 0$ and $\beta$. Define
	\begin{equation}
		\label{eq:rho-u-def}
		\rho=|\psi|^2,
		\qquad
		u=\sigma\,\operatorname{Im}\left(\frac{\psi_s}{\psi}\right).			
	\end{equation}
	Equivalently, if $\psi=\sqrt{\rho}e^{i\theta}$, then \( u=\sigma\theta_s. \)
	Then $(\rho,u)$ satisfies the Betchov-type conservation laws
\begin{align}
	\rho_t+J_s &=0,
	\qquad
	J(\rho,u)=\rho u,
	\label{eq:convlaw1}
	\\
	J_t+K_s &=0,
	\qquad
	K(\rho,u)
	=
	\rho u^2
	-\frac{\sigma\beta}{2}\rho^2
	-\frac{\sigma^2}{4}\rho(\log\rho)_{ss}.
	\label{eq:convlaw2}
\end{align}
\end{theorem}
\begin{rmk}
	The formulation above is pointwise valid on regions where $\rho>0$. If
	$\psi$ has zeros, then $u=\sigma\operatorname{Im}(\psi_s/\psi)$ and
	$(\log\rho)_{ss}$ may be singular. In such cases the continuity equation is
	better written in current form,
	\[
	\rho_t+J_s=0,
	\qquad
	J=\sigma\,\operatorname{Im}(\psi^*\psi_s),
	\]
	where $\psi^*$ denotes the complex conjugate of $\psi$; the velocity formulation is understood away from vacuum points.
\end{rmk}
\begin{rmk}
	For the more general equation
	\begin{equation}
		\label{eq:ivancevic-nls-gen}
		i\psi_t+D\psi_{ss}+\beta|\psi|^2\psi=0,		
	\end{equation}
	with $u=2D\,\theta_s$, the second flux becomes
	\begin{equation}
		\label{eq:flux-gen}
		K(\rho, u) = \rho u^2-D\beta\rho^2-D^2\rho(\log\rho)_{ss}.		
	\end{equation}
	The Ivancevic model \eqref{eq:ivancevic-nls} corresponds to $D=\sigma/2$.		
\end{rmk}
\begin{rmk}
	The constant-coefficient assumption is important for the closed conservation-law structure derived above. The Madelung decomposition itself continues to apply for sufficiently regular real variable coefficients in the region where $\rho>0$, but coefficient variations generally introduce additional terms into the hydrodynamic equations. The resulting
	system therefore takes the form of a balance-law system rather than the
	closed Betchov-type conservation laws \eqref{eq:convlaw1}--\eqref{eq:convlaw2}.
	The precise form of these additional terms depends on whether the coefficients
	vary in space, time, or both. A systematic treatment of the variable- and
	stochastic-coefficient cases is beyond the scope of the present work.
\end{rmk}

Although Ivancevic's stochastic-volatility formulation corresponds to the two-component case \(n=2\), the same calculation applies to the natural \(n\)-component Manakov generalization under a fixed-polarization assumption.  
\begin{cor}[Fixed-polarization vector extension]
	\label{cor:thm1}
	Let
	\(q(s,t)= (q_1(s,t), \ldots, q_n(s,t))^T\)
	solve the $n$-component Ivancevic--Manakov system
	\begin{equation}
		\label{eq:Manakov}
	i q_t+Dq_{ss}+\beta(q^\dagger q)q=0.		
	\end{equation}
	Assume that $q$ has fixed polarization, namely
	\[
	q(s,t)=c\,\sqrt{\rho(s,t)}\,e^{i\theta(s,t)},
	\qquad
	c^\dagger c=1,
	\]
	where $c\in\mathbb C^n$ is constant and $c^\dagger$ denotes the conjugate transpose. Define
	\[
	\rho=q^\dagger q=\sum_{j=1}^n |q_j|^2,
	\qquad
	u=2D\,\theta_s.
	\]
	Then $(\rho,u)$ satisfies \eqref{eq:convlaw1}, and \eqref{eq:convlaw2} with $K(\rho, u) = \rho u^2-D\beta\rho^2 -D^2\rho(\log\rho)_{ss}$. 
	As before, the Ivancevic scaling is $D=\sigma/2$.
\end{cor}

\begin{proof}[Proof of Theorem~\ref{thm:scalar-betchov}]
	Let
	\[
	\psi=\sqrt{\rho}\,e^{i\theta},
	\qquad
	\rho>0.
	\]
	Substituting this expression into
	\[
	i\psi_t+\frac{\sigma}{2}\psi_{ss}
	+\beta|\psi|^2\psi=0
	\]
	and separating real and imaginary parts gives the two equations
	\begin{equation}
	\rho_t+\sigma(\rho\theta_s)_s=0,
	\label{eq:proof-continuity-theta}		
	\end{equation}
	and
	\begin{equation}
	\theta_t+\frac{\sigma}{2}\theta_s^2
-\frac{\sigma}{2}\frac{(\sqrt{\rho})_{ss}}{\sqrt{\rho}}
-\beta\rho=0.
\label{eq:proof-phase}		
	\end{equation}
By defining $u=\sigma\theta_s$, \eqref{eq:proof-continuity-theta} becomes
	\[
	\rho_t+(\rho u)_s=0.
	\]
	
Next, differentiating
\eqref{eq:proof-phase} with respect to $s$, we obtain
\begin{equation}
	u_t+u u_s
-\frac{\sigma^2}{2}
\partial_s
\left(
\frac{(\sqrt{\rho})_{ss}}{\sqrt{\rho}}
\right)
-\sigma\beta\rho_s=0.
\label{eq:proof-velocity}		
	\end{equation}
Multiplying \eqref{eq:proof-velocity} by \(\rho\), using the continuity equation, and rewriting the dispersive term with the identity
\[
\rho\partial_s
\left(
\frac{(\sqrt{\rho})_{ss}}{\sqrt{\rho}}
\right)
=
\frac12\partial_s\left[\rho(\log\rho)_{ss}\right],
\]
valid for smooth positive \(\rho\), yields the conservative form
\[
(\rho u)_t+
\partial_s
\left[
\rho u^2
-\frac{\sigma\beta}{2}\rho^2
-\frac{\sigma^2}{4}\rho(\log\rho)_{ss}
\right]=0.
\]
\end{proof}
\begin{proof}[Proof of Corollary~\ref{cor:thm1}]
	By fixed polarization we mean that the polarization vector $c\in \mathbb{C}^n$ is constant, independent of $s$ and $t$, with $c^\dagger c= 1$; thus $q=c \sqrt{\rho}e^{i\theta}$. Then $q^\dagger q= |\psi|^2$ with $\psi = \sqrt{\rho} e^{i\theta}$, and the Manakov system \eqref{eq:Manakov} reduces to 		\eqref{eq:ivancevic-nls-gen}. The result therefore follows from Theorem~\ref{thm:scalar-betchov}; the remaining calculations are identical and are omitted. 
\end{proof}		

\subsection{Explicit examples}
In this section we illustrate the Betchov-type formulation on several explicit solutions of the Ivancevic model \eqref{eq:ivancevic-nls} considered in the literature; for instance, the scalar Ivancevic solitons, the Manakov two-component soliton, and the multi-soliton density obtained through Hirota's formalism.

\subsubsection{Scalar Ivancevic bright soliton}
\label{subsec:brightsolitonexample}
A standard bright soliton solution of \eqref{eq:ivancevic-nls} is
\begin{equation}
	\psi(s,t) =
	A\,\sech(s-\sigma kt)
	\exp\left[
	i\left(
	ks-\frac{\sigma}{2}(k^2-1)t
	\right)
	\right], \quad \text{with} \quad 		A^2=\frac{\sigma}{\beta},
\end{equation}
where \(k\in\mathbb{R}\), \(\sigma>0\), and \(\beta>0\).
Thus, from \eqref{eq:rho-u-def}, we have
\begin{equation}
	\rho(s,t)
	=
	\frac{\sigma}{\beta}
	\sech^2(s-\sigma kt),
	\qquad
	u(s,t)=\sigma k.
\end{equation}
Since $\rho=\rho(s-\sigma kt)$ and $u=\sigma k$, the continuity equation \eqref{eq:convlaw1} is immediate. Furthermore, by writing
\begin{equation}
	(\log\rho)_{ss}
	=
	-2\sech^2(s-\sigma kt)
	=
	-\frac{2\beta}{\sigma}\rho,
\end{equation}
the nonlinear pressure and dispersive terms in flux $K(\rho, u)$ cancel and \eqref{eq:convlaw2} follows.

\subsubsection{Scalar Ivancevic dark soliton}

The dark or shock-wave solution has the form
\begin{equation}
	\psi(s,t)
	=
	A\,\tanh(s-\sigma kt)
	\exp\left[
	i\left(
	ks-\frac{\sigma}{2}(k^2+2)t
	\right)
	\right].
\end{equation}
For the normalization used here, the dark soliton corresponds to the defocusing sign and we take $A^2=-\sigma/\beta$, and obtain
\begin{equation}
	\rho(s,t)
	=
	-\frac{\sigma}{\beta}
	\tanh^2(s-\sigma kt),
	\qquad
	u(s,t)=\sigma k.
\end{equation}
The first Betchov equation is again satisfied because $\rho=\rho(s-\sigma kt)$. The verification of the second equation is pointwise away from the zero of $\rho$, i.e., away from $s-\sigma kt=0$. In this case, the logarithmic term $(\log\rho)_{ss}$ is singular at the density zero, so the hydrodynamic formulation should be understood on the open set where $\rho>0$. 

\begin{rmk}
	The dark-soliton example should be understood as the defocusing counterpart of
	the bright-soliton case. For \(\sigma>0\), this corresponds to \(\beta<0\) in
	the Ivancevic equation. This focusing/defocusing distinction also appears in
	the geometric theory of curve flows: the classical Euclidean Hasimoto
	correspondence leads to the focusing NLS, while related curve flows in
	hyperbolic settings lead to defocusing NLS-type equations \cite{Ding1998,deLaHozKumarVega2022,Kumar2022}.
	We only use this observation as a geometric analogy; the Betchov-type formulation
	above is derived directly from the Ivancevic NLS.
\end{rmk}

\subsubsection{Two-component Manakov solution}
A coupled volatility--option-price system of Manakov type was considered in
\cite{Ivancevic2010}; see also the original Manakov system \cite{Manakov1974}
and related geometric interpretations of the Manakov system in terms of
interacting moving curves \cite{Kostov2007manakov}.

This example corresponds to the two-component case of Corollary \ref{cor:thm1}. Let

\begin{equation}
	q(s,t)
	=
	\begin{pmatrix}
		\sigma_f(s,t)\\
		\psi(s,t)
	\end{pmatrix},
\end{equation}
where $\sigma_f$ denotes the volatility wave and $\psi$ denotes the option-price
wave. The corresponding two-component NLS has the schematic form
\begin{equation}
	i q_t+Dq_{ss}+\beta(q^\dagger q)q=0,
\end{equation}
where
\begin{equation}
	q^\dagger q=|\sigma_f|^2+|\psi|^2.
\end{equation}
For this vector system the natural density variable is the total intensity
\begin{equation}
	\rho=q^\dagger q=|\sigma_f|^2+|\psi|^2,
\end{equation}
and the corresponding velocity is
\begin{equation}
	u
	=
	2D
	\frac{
		\operatorname{Im}(q^\dagger q_s)
	}{
		q^\dagger q
	}.
\end{equation}
Ivancevic's two-component bright soliton can be written as
\begin{equation}
	q(s,t)
	=
	2b\,\mathbf c\,
	\sech\bigl(2b(s+4at)\bigr)
	\exp\left[
	-2i(2a^2t+as-2b^2t)
	\right], 
\end{equation}
where $\mathbf c= \begin{pmatrix}
	c_1, 
	c_2
\end{pmatrix}^T$, $|c_1|^2+|c_2|^2=1.$

Therefore
\begin{equation}
	\rho(s,t)
	=
	q^\dagger q
	=
	4b^2\sech^2\bigl(2b(s+4at)\bigr),
\end{equation}
and the phase
\begin{equation}
	\theta(s,t)
	=
	-2(2a^2t+as-2b^2t),
\end{equation}
implies
\begin{equation}
	u=2D\theta_s=-4aD.
\end{equation}

This example shows that in the coupled Ivancevic model the natural density variable is not the density of one component alone, but the total intensity of the vector wave. Together with the common phase-gradient velocity in the fixed-polarization case, this total density gives the Betchov-type hydrodynamic formulation described in Corollary~\ref{cor:thm1}. The same observation applies to the natural \(n\)-component Manakov generalization under the same fixed-polarization assumption.

\subsubsection{$N$-bright-soliton densities}
\label{subsec:Nbrightsolitonexample}
In \cite{Delisle2026}, the authors construct bright $N$-soliton solutions of the
scalar Ivancevic equation by Hirota's bilinear method \cite{HirotaNLS1973}. The solution has the form
\begin{equation}\label{eq:psi_Nmulti}
	\psi_N(s,t)=\frac{f_N(s,t)}{g_N(s,t)}\,\psi_0(s,t),
\end{equation}
where
\begin{equation}\label{eq:psi_0multi}
\psi_0(s,t)
=
\exp(i(\kappa s+\omega t+\phi)),
\end{equation}
$f_N$ is complex-valued, $g_N$ is real-valued. Here \(\kappa \in \mathbb{R}\) is the carrier wavenumber, \(\phi\) is a constant phase, and $\omega=\beta-\frac{\sigma}{2}\kappa^2$, the corresponding carrier frequency.  

The corresponding density is
\begin{equation}
	\rho_N=|\psi_N|^2=\left|\frac{f_N}{g_N}\right|^2=\frac{|f_N|^2}{g_N^2};
\end{equation}
using the Hirota bilinear identity in the normalization adopted by Delisle et al.~\cite{Delisle2026}, this becomes
\[
\rho_N
=
\frac{\sigma}{\beta}
(\log g_N)_{ss},
\]
which reduces, for $N=1$, to the $\sech^2$ profile of Subsection~\ref{subsec:brightsolitonexample}.  The velocity field is
\begin{equation}
	u_N(s,t)=\sigma\,\operatorname{Im}\!\left(\frac{(\psi_N)_s}{\psi_N}\right)
	=\sigma\!\left[\kappa+\operatorname{Im}\!\left(\frac{(f_N)_s}{f_N}\right)\right],
\end{equation}
where the second equality follows from
\eqref{eq:psi_Nmulti}--\eqref{eq:psi_0multi}. By the general Betchov-type
formulation, the pair $(\rho_N,u_N)$ satisfies
\begin{equation}
	(\rho_N)_t+(\rho_N u_N)_s=0
\end{equation}
and
\begin{equation}
	(\rho_N u_N)_t+\partial_s\!\left[\rho_N u_N^2-\frac{\sigma\beta}{2}\rho_N^2
	-\frac{\sigma^2}{4}\rho_N(\log\rho_N)_{ss}\right]=0,
\end{equation}
wherever $\rho_N>0$.

As mentioned earlier, for $N=1$, the density reduces to the localized $\sech^2$ profile of the bright-soliton example, and for $N=2$, $\rho_2$ is a single density field with two interacting localized peaks; it is not the sum of two independent one-soliton densities. In order to illustrate this, with parameters and expressions given in Appendix \ref{app:twosoliton}, in Figure \ref{fig:hirota-densities}\subref{fig:rho2soliton}, we have plotted the time evolution of $\rho(s,t)$ for $0\le t\le40$. The figure shows two initially separated peaks evolving so that one remains approximately stationary while the other propagates through it. On the right, Figure \ref{fig:hirota-densities}\subref{fig:J2soliton} shows the corresponding probability current $J_2=\rho_2u_2$. Unlike the density, the probability current describes the transport of probability mass. Consistent with the evolution of $\rho_2$, $J_2$ is concentrated primarily along the moving branch, since the approximately stationary density peak contributes little to the current. The full two-soliton field should therefore be interpreted as a single nonlinear solution rather than a superposition of two independent one-soliton states. Likewise, for general $N$, $\rho_N$ represents a single probability density with $N$ interacting localized regions.

\begin{figure}[t]
	\centering
	\begin{subfigure}{0.48\textwidth}
		\centering
		\includegraphics[width=\textwidth]{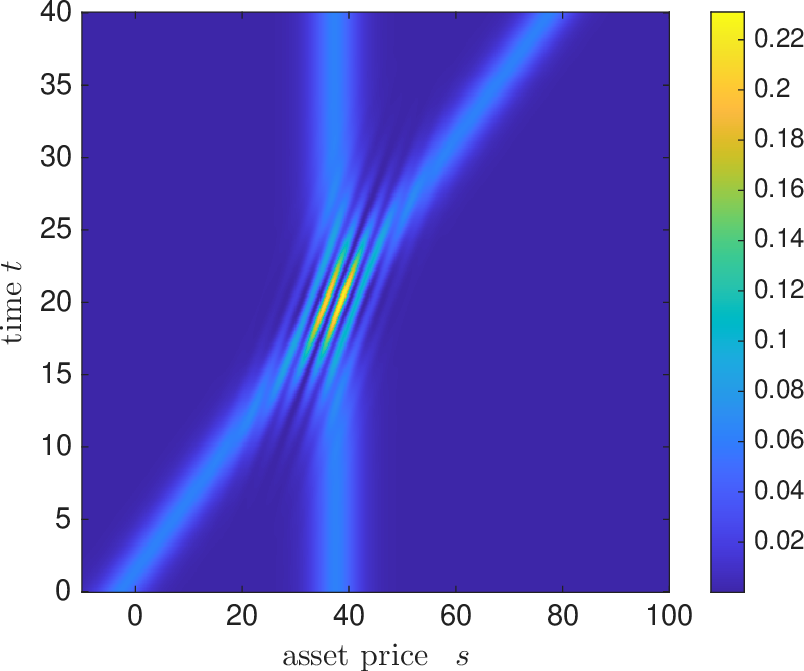}
		\caption{\(\rho_2(s,t)\)}	
		\label{fig:rho2soliton}	
	\end{subfigure}
		\hfill
	\begin{subfigure}{0.48\textwidth}
		\centering
		\includegraphics[width=\textwidth]{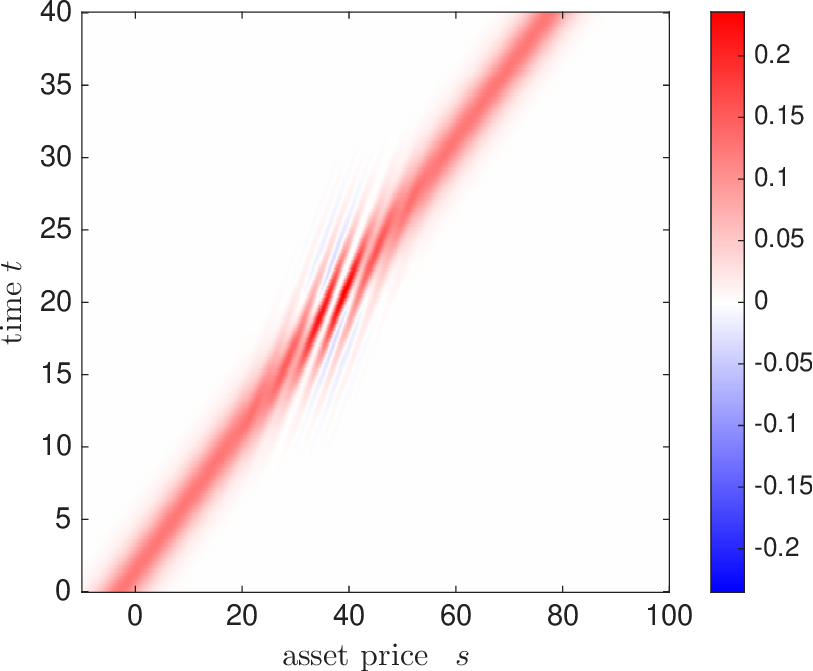}
		\caption{\(J(s,t)\)}		
		\label{fig:J2soliton}
	\end{subfigure} 
\caption{Density $\rho_2=|\psi_2|^2$ (left) and probability current
	$J_2=\sigma\operatorname{Im}(\psi_2^*(\psi_2)_s)$ (right) for the
	bright two-soliton solution of the Ivancevic equation. One branch is nearly stationary while the other propagates through it, so the probability current is concentrated mainly along
	the moving branch. The corresponding parameter values and expressions are given in Appendix~\ref{app:twosoliton}. }
	\label{fig:hirota-densities}
\end{figure}

\begin{rmk}[Normalization and price domain]
	\label{rmk:normalization}
For a probabilistic interpretation, the density $\rho=|\psi|^2$ should be
	nonnegative, integrable and normalized to unit mass. This is the case for the
	localized bright and Manakov solutions considered in this paper. Since the
	continuity equation \eqref{eq:convlaw1} conserves the total mass
	$\int_{\mathbb R}\rho(s,t)\,ds$, one may normalize
	\[
	\widehat{\rho}(s,t)
	=
	\frac{\rho(s,t)}
	{\int_{\mathbb R}\rho(\xi,t)\,d\xi},
	\]
	so that
	\[
	\int_{\mathbb R}\widehat{\rho}(s,t)\,ds=1,
	\]
	and
	\[
	\int_a^b\widehat{\rho}(s,t)\,ds
	\]
	gives the probability assigned to the price interval $[a,b]$.

	Although the mathematical solutions are naturally posed on the whole real
	line, the calibration of Subsection~\ref{subsec:calibration} uses parameter values
	for which the normalized density is concentrated on the positive half-line,
	making the interpretation in terms of asset prices natural. By contrast, dark
	solitons approach a nonzero background and are therefore not integrable on
	$\mathbb R$; they should consequently be viewed as mathematical examples of the
	defocusing Betchov-type formulation rather than as normalized probability
	densities.
\end{rmk}

\section{Financial interpretation} \label{sec:Section3}
The results presented in the previous sections allow us to give a hydrodynamic interpretation of the Ivancevic option-pricing wave function. In the original Ivancevic model \cite{Ivancevic2010}, the unknown complex-valued function $\psi(s, t)$ denotes an option-price wave, while $\rho=|\psi(s,t)|^2$ can be understood as a probability density in the asset-price variable $s$. Thus, when $\rho$ is normalized as in Remark \ref{rmk:normalization}, for a given time $t$, the quantity 
\begin{equation}
	\int_a^b \rho(s,t)\,ds
\end{equation}
may be interpreted as the probability mass assigned to the price interval $[a, b]$, 
and
\begin{equation}
	u(s,t)
	=
	\sigma\,\operatorname{Im}\left(\frac{\psi_s}{\psi}\right),
\end{equation}
is the phase-gradient velocity associated with the option-price wave.

Indeed, the first equation \eqref{eq:convlaw1} is a continuity equation and it states that changes in the option-price density are driven by the divergence of the probability current
\begin{equation}
	J=\rho u
	=
	\sigma\,\operatorname{Im}(\psi^*\psi_s).
\end{equation}
The current $J$ represents the model-implied amount of probability mass flowing past a given price level per unit time.  With this interpretation, at each price level $s$, $u$ determines the direction and speed of transport of the density, i.e., $u>0$ implies transport toward higher price levels, while $u<0$ corresponds to transport toward lower price levels.

This interpretation complements, rather than replaces, the classical Greeks (Delta, Gamma and Theta). The Greeks measure the sensitivity of the option value, i.e., $V$ in \eqref{eq:black-scholes}, to variables such as the asset price, time, volatility and the interest rate. By contrast, $u$ is a phase-gradient transport variable. For example, the Delta of a classical option measures sensitivity with respect to the asset price, whereas $u$ describes the local phase-driven propagation of the probability density along the asset-price axis. Similarly, while Ivancevic identifies $\rho_t$ with Theta, this quantity measures the redistribution of probability mass, rather than the time decay of an option price in the classical sense. 

The momentum-type equation \eqref{eq:convlaw2} provides a second layer of interpretation. The first term $\rho u^2$ in the flux $K$, is the convective transport of momentum, while the second term $(-\sigma\beta \rho^2/2)$ is the nonlinear pressure term induced by the cubic nonlinearity of the Ivancevic model \eqref{eq:ivancevic-nls}. For $\beta>0$, corresponding to the focusing case, this pressure contribution is negative, in analogy with the focusing system arising from vortex filament dynamics. If $\beta<0$, the sign reverses and the model has a defocusing or positive-pressure character.  Finally, the term $(-\sigma^2 \rho(\log\rho)_{ss}/4)$ is the dispersive, or quantum-pressure, contribution. This term encodes the effect of spatial variations of the density and is responsible for the regularizing dispersive behavior characteristic of nonlinear Schrödinger dynamics.

With respect to the examples discussed above, for a bright soliton, the density $\rho$ is localized and transported with constant velocity. Within the model, such a profile represents a localized region of high probability in the asset price variable. On the other hand, for a dark soliton, the density contains a dip on a nonzero background. In that case, the logarithmic formulation must be understood away from points where $\rho=0$. Within the model, such points correspond to asset-price levels of zero model density.

In the two-component Manakov formulation of the Ivancevic model, the natural density is not the density of one component alone, but the total intensity
\begin{equation}
	\rho=|\sigma_f|^2+|\psi|^2,
\end{equation}
where $\sigma_f$ denotes the volatility wave and $\psi$ denotes the option-price wave. In our analogy, this quantity measures the total intensity of the coupled volatility--option-price system. When restricted to the fixed-polarization case, i.e., both components share a common phase gradient, the associated velocity reduces to this common phase-gradient velocity, and the total density satisfies the same scalar Betchov-type conservation laws. This provides
a natural interpretation of Ivancevic's coupled model; that is, the joint market-wave intensity is transported along the asset-price axis at a velocity determined by the shared phase of the volatility and option-price waves.

Finally, for the scalar \(N\)-soliton solutions, \(\rho_N=|\psi_N|^2\) represents a single probability density with \(N\) interacting localized peaks. These peaks may be interpreted, within the model, as localized asset-price regions with non-negligible probability. The associated current $J_N$ describes how this probability mass is transported along the asset-price axis, as illustrated in Figure~\ref{fig:hirota-densities} for \(N=2\). During nonlinear interaction, the density is redistributed through interference effects, and after interaction the individual soliton profiles are asymptotically restored, as is characteristic of elastic soliton collisions.
 
\subsection{Static density calibration}
\label{subsec:calibration}

To illustrate the density interpretation concretely, we compare the
normalized one-soliton density with a standard probability density arising in
option pricing. Since $\rho=|\psi|^{2}$ is interpreted within the Ivancevic
model as a probability density over asset prices, the closest standard market
object is the \emph{risk-neutral density} (RND), namely the probability density
under the pricing measure used to value European options. Unlike implied
volatility, which parametrizes option prices only indirectly, the risk-neutral
density is itself a probability density over future asset prices. We therefore
use the Black--Scholes risk-neutral density as a synthetic benchmark for the
model density $\rho$.

The Black--Scholes risk-neutral density is the probability density of the
underlying asset price at maturity time $T$ under the risk-neutral probability
measure $\mathbb Q$ used for arbitrage-free option pricing. Here $S_0$ denotes
the current asset price, $r$ the continuously compounded risk-free interest
rate, and $\sigma_{BS}$ the Black--Scholes volatility. This density is recoverable from
call prices $C(K)$ at strike price $K$ through the Breeden--Litzenberger
relation \cite{BreedenLitzenberger1978},
\[
q(K)=e^{rT} C_{KK},
\]
and, under the Black--Scholes model, is given explicitly by the lognormal
density
\begin{equation}
	q_{BS}(s)
	=
	\frac{1}{s\,\Sigma\sqrt{2\pi}}
	\exp\!\left(
	-\frac{(\ln s-\mu)^2}{2\Sigma^2}
	\right),
	\qquad
	\mu=\ln S_0+\left(r-\tfrac12\sigma_{BS}^{2}\right)T,
	\qquad
	\Sigma=\sigma_{BS}\sqrt{T},
\end{equation}
where
\begin{align*}
	C(K)
	&=
	S_0\Phi(d_1)-Ke^{-rT}\Phi(d_2),\\
	d_1
	&=
	\frac{\ln(S_0/K)+(r+\tfrac12\sigma_{BS}^{2})T}
	{\sigma_{BS}\sqrt{T}},
	\qquad
	d_2
	=
	\frac{\ln(S_0/K)+(r-\tfrac12\sigma_{BS}^{2})T}
	{\sigma_{BS}\sqrt{T}},
\end{align*}
and $\Phi$ denotes the standard normal distribution function.

Within the present interpretation, the risk-neutral density provides the
closest standard market benchmark for $\rho$. The phase-derived quantities
$u$ and $J$, however, are not directly observable from a static option-price
snapshot. The static calibration considered below therefore uses only the density
$\rho=|\psi|^2$ and does not, by itself, require the full hydrodynamic
reformulation. Its purpose is to test the probabilistic interpretation of
$\rho$ against a standard financial benchmark and to make the dependence of
the Ivancevic coefficient ratio on market quantities explicit.

The normalized one-soliton density is
\[
\widehat{\rho}(s)
=
\frac{\eta}{2}
\operatorname{sech}^{2}\!\bigl(\eta(s-s_0)\bigr),
\]
depending only on the location $s_0$ and inverse width $\eta$, with the
unit-mass condition fixing (see Appendix~\ref{app:bright})
\[
\frac{\sigma}{\beta}
=
\frac{1}{2\eta}.
\] Since both
$\widehat{\rho}$ and $q_{BS}$ are normalized probability densities on the asset price variable, it is natural to compare them by least squares fitting. For parameters
$S_0=100$, $r=0.03$, $\sigma_{BS}=0.20$, and $T=0.5$, we determine the optimal
$s_0$ and $\eta$ by minimizing
\begin{equation}
	\label{eq:leastsquaredef}
	\int_{0}^{\infty}
	\left(
	\widehat{\rho}(s)-q_{BS}(s)
	\right)^2\,ds.
\end{equation}
The objective function \eqref{eq:leastsquaredef} is minimized using the MATLAB routine \texttt{fminsearch}, initialized at the moment-matching values of Proposition~\ref{prop:implied}, with the integral approximated by the \texttt{trapz} function on a discretized price grid.  With this, we obtain a reasonably good approximation of the central probability mass, shown in Figure~\ref{fig:calibration}\subref{fig:calib}, with relative
$L^2$ error approximately $8\%$. The symmetric $\operatorname{sech}^{2}$ profile captures the central part of the lognormal density but not its skew, consistent with a single bright soliton being a deliberately simple model density. 

In fact, the calibration also makes the coefficient dependence explicit. Let
\[
F
=
S_0e^{rT}
=
\mathbb E^{\mathbb Q}[S_T]
=
\int_{0}^{\infty}s\,q_{BS}(s)\,ds
\]
be the forward price, and let
\[
s_{\mathrm{RND}}^{2}
=
\operatorname{Var}(S_T)
=
F^{2}(e^{\Sigma^{2}}-1)
\]
be the variance of the risk-neutral density.  With this, we first derive a closed-form moment-matching relation, which provides an analytic benchmark for the calibration.
\begin{prop}[Implied coefficient ratio]
	\label{prop:implied}
	Let
	\begin{equation}
		\label{eq:rhohatdef}
	\widehat{\rho}(s)
	= \frac{\eta}{2}
	\operatorname{sech}^{2}\!\bigl(\eta(s-s_0)\bigr),
	\qquad \eta>0,		
	\end{equation}
	be the normalized one-soliton density of the Ivancevic model, so that the unit-mass condition gives $\sigma/\beta=1/(2\eta)$. Suppose that the first two moments of $\widehat{\rho}$ are matched to those of a benchmark probability density with mean $F$ and variance
	$s_{\mathrm{RND}}^{2}>0$. Then
	\[
	s_0=F,
	\qquad
	\eta=\frac{\pi}{2\sqrt{3}\,s_{\mathrm{RND}}},
	\qquad
	\frac{\sigma}{\beta}
	=
	\frac{\sqrt{3}}{\pi}\,s_{\mathrm{RND}}.
	\]
	In particular, for the Black--Scholes risk-neutral density (with $S_0$ the current asset price, $r$ the risk-free rate, $T$ the maturity, and $\sigma_{BS}$ the volatility),
	\[
	F=S_0e^{rT},
	\qquad
	s_{\mathrm{RND}}^{2}
	=
	F^{2}\bigl(e^{\Sigma^{2}}-1\bigr),
	\qquad
	\Sigma=\sigma_{BS}\sqrt{T},
	\]
	and hence
	\begin{equation}
			\label{eq:implied-sigbetaexact}
	\frac{\sigma}{\beta}
	= \frac{\sqrt{3}}{\pi}
	F\sqrt{e^{\Sigma^{2}}-1}.
	\end{equation} 
	Moreover, in the small-volatility--maturity regime $\Sigma\ll1$,
	\begin{equation}
	\label{eq:implied-sigbeta}
		\frac{\sigma}{\beta}
		=
		\frac{\sqrt{3}}{\pi}
		\sigma_{BS}\sqrt{T}S_0e^{rT}
		+
		O\!\left(\Sigma^{3}\right).
	\end{equation}
\end{prop}
\begin{proof}
	By the change of variables $y=\eta(s-s_0)$ and the symmetry of
	$\sech^2 y$, the first moment is
	\[
	\int_{\mathbb R}s\widehat{\rho}(s)\,ds=\frac{1}{2}\int_{\mathbb R}\left(\frac{y}{\eta}+s_0\right)\sech^2(y)\,dy=s_0.
	\]
	Moreover,
	\[
	\int_{\mathbb R}(s-s_0)^2\widehat{\rho}(s)\,ds
	=
	\frac{1}{2\eta^2}
	\int_{\mathbb R}y^2\sech^2 y\,dy
	=
	\frac{\pi^2}{12\eta^2}.
	\]	
	Matching these moments to $F$ and $s_{\mathrm{RND}}^{2}$ gives
	\[
	s_0=F,
	\qquad
	\eta=\frac{\pi}{2\sqrt{3}\,s_{\mathrm{RND}}}.
	\]
	Using the unit-mass identity $\sigma/\beta=1/(2\eta)$, we obtain
	\[
	\frac{\sigma}{\beta}
	=
	\frac{\sqrt{3}}{\pi}\,s_{\mathrm{RND}}.
	\]
	For the Black--Scholes risk-neutral density, substitution of \( s_{\mathrm{RND}} = F\sqrt{e^{\Sigma^{2}}-1}\)
	gives the exact relation \eqref{eq:implied-sigbetaexact}. Finally,
	\[
	\sqrt{e^{\Sigma^{2}}-1}
	=
	\Sigma+O(\Sigma^{3})
	\qquad
	\text{as} \quad \Sigma\to0,
	\]
	which yields \eqref{eq:implied-sigbeta}.
\end{proof}
Thus, we have an explicit analytical relationship between the coefficient ratio $\sigma/\beta$ and the first two moments of the benchmark density, which in the Black–Scholes case, yields \eqref{eq:implied-sigbeta}. To assess this approximation numerically, we repeat the least-squares fit for $\sigma_{BS}\in\{0.15,0.20,0.30\}$ and
$T\in\{0.25,0.5,1,2\}$, fixing
$S_0=100$ and $r=0.03$. For each pair
$(\sigma_{BS},T)$ we construct the corresponding Black--Scholes risk-neutral
density and determine the optimal parameters $(s_0,\eta)$ by minimizing
\eqref{eq:leastsquaredef}. The fitted width then gives the implied ratio
$\sigma/\beta=1/(2\eta)$.

\begin{figure}[h]
	\centering
	\begin{subfigure}{0.49\textwidth}
		\centering
		\includegraphics[width=\textwidth]{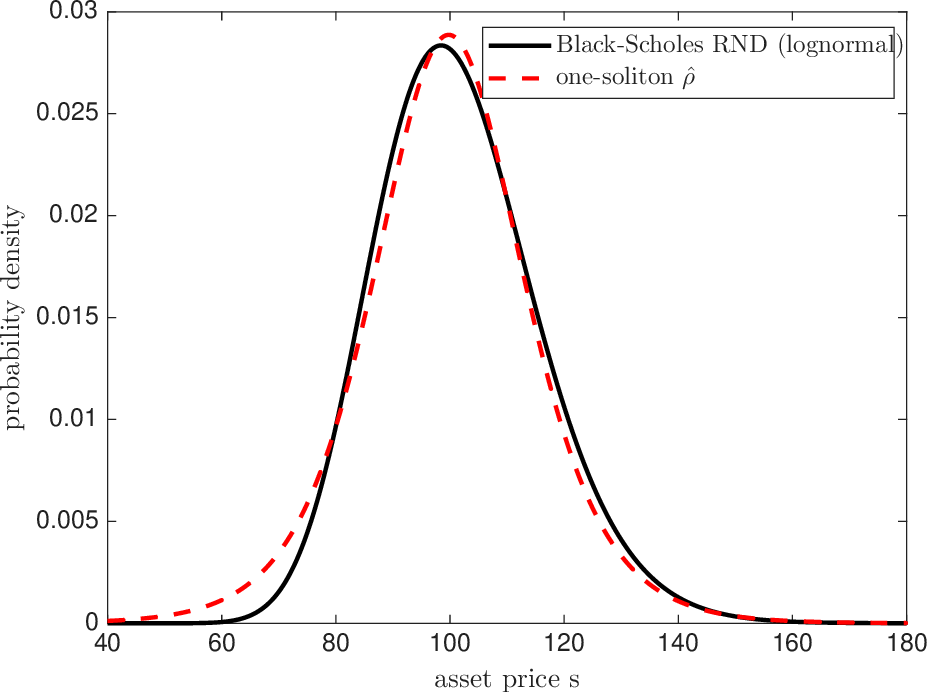}
		\caption{}\label{fig:calib}
	\end{subfigure}
	\hfill
	\begin{subfigure}{0.495\textwidth}
		\centering
		\includegraphics[width=\textwidth]{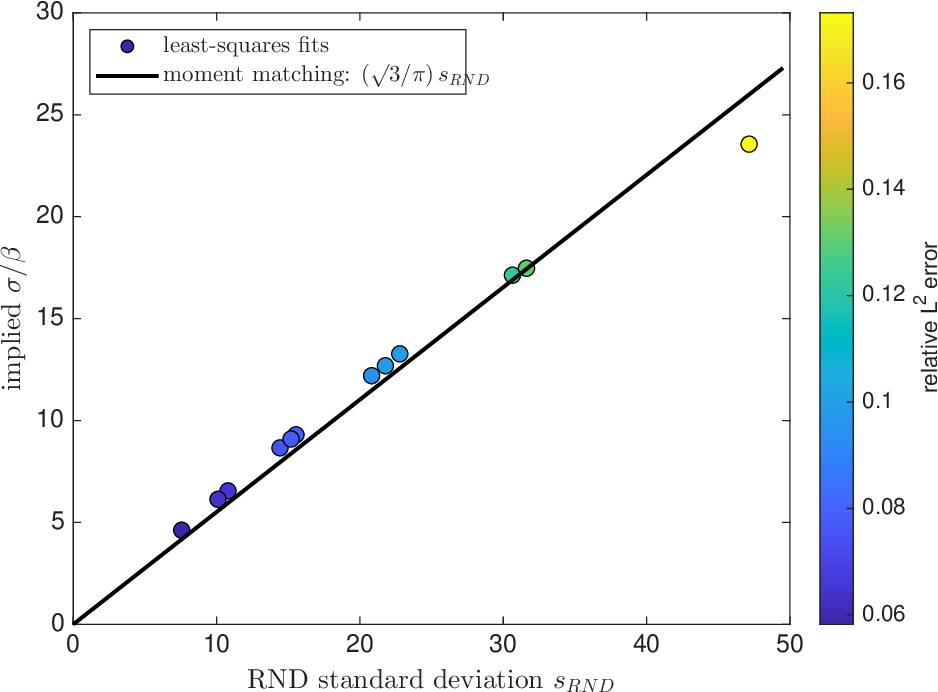}
		\caption{}\label{fig:implied}
	\end{subfigure}
	\caption{Illustrative one-soliton calibration. (A) Least-squares fit of the
		normalized one-soliton density $\widehat{\rho}$ (dashed) to the Black--Scholes
		risk-neutral density $q_{BS}$ (solid). (B) Implied ratio $\sigma/\beta$
		obtained from least-squares fits over several pairs
		$(\sigma_{BS},T)$ (markers), compared with the theoretical relation
		\eqref{eq:implied-sigbeta} (solid line). Marker colors indicate the relative
		$L^2$ fitting error.}
	\label{fig:calibration}
\end{figure}

Figure~\ref{fig:calibration}\subref{fig:implied} shows that the least-squares estimates collapse
closely onto the moment-matching prediction
\eqref{eq:implied-sigbeta}, lying approximately $6\%$ above it over the
parameter range considered. At the same time, the relative fitting error grows
from about $6\%$ for short maturities to about $20\%$ for longer maturities, as
the skew of the lognormal risk-neutral density becomes more pronounced.

Although synthetic, this example is illustrative. Unlike the numerical Black–Scholes fits considered in the original adaptive-wave model \cite{Ivancevic2010,Ivancevic2011}, the present comparison yields an explicit analytical relation between the Ivancevic coefficient ratio and standard market quantities. 


\subsection{Dynamical reconstruction of the probability current}
\label{subsec:dynamiccurrent}
The calibration considered in the previous subsection is intentionally
static and uses only a single snapshot of the density. When the temporal
evolution of the density is available, the conservation-law formulation
provides additional information through the probability current.

Indeed, although the probability current
\[
J=\sigma\,\operatorname{Im}(\psi^*\psi_s)
\]
depends on the phase of the complex wave function, the continuity equation
\[
\rho_t+J_s=0
\]
shows that the current can be reconstructed from the temporal evolution of
the density, up to a boundary value.

Integrating the continuity equation from $s_{\min}$ to $s$ gives
\[
J(s,t)= J(s_{\min},t) - \int_{s_{\min}}^s \rho_t(\xi,t)\,d\xi.
\]
For the localized solutions considered here, the computational boundary is
chosen sufficiently far into the tail so that $J(s_{\min},t)\approx0$. We
therefore define the density-reconstructed current
\begin{equation}
	\label{eq:Jrec}
	J_{\mathrm{rec}}(s,t)
	=
	-\int_{s_{\min}}^s \rho_t(\xi,t)\,d\xi.
\end{equation}

We test \eqref{eq:Jrec} using the analytic bright two-soliton solution of
Subsection~\ref{subsec:Nbrightsolitonexample}, with the parameters given in
Appendix~\ref{app:twosoliton}. We take $s_{\min}=-120$, the left boundary of
the computational domain $[-120,140]$, where the current is numerically
negligible. The density
\[
\rho_2(s,t)=|\psi_2(s,t)|^2
\]
is sampled on the same space--time grid as in Figure~\ref{fig:hirota-densities}.
The time derivative $\rho_t$ is approximated by fourth-order centered finite
differences, and the integral in \eqref{eq:Jrec} is evaluated by cumulative
trapezoidal quadrature. The reconstructed current is then compared with the
current calculated directly from the complex field,
\[
J_\psi(s,t)
=
\sigma\,\operatorname{Im}
\bigl(\psi_2^*(\psi_2)_s\bigr).
\]

We measure the reconstruction accuracy using
\begin{equation}
	\label{eq:EJ}
	E_J(t)
	=
	\frac{
		\|J_{\mathrm{rec}}(\cdot,t)-J_\psi(\cdot,t)\|_{L^2}
	}{
		\|J_\psi(\cdot,t)\|_{L^2}
	}.
\end{equation}

In Figure~\ref{fig:dynamicdensity}A, we plot $J_\psi(\cdot,t)$ and $J_{\mathrm{rec}}(\cdot,t)$ at $t=20$, while the relative error $E_J(t)$ is shown in Figure~\ref{fig:dynamicdensity}B for $t\in[0,40]$. The reconstruction remains accurate throughout the evolution, with the maximum error occurring during the nonlinear interaction near $t=20$, i.e., \(E_J(20)\approx 4.71\times10^{-4}\).

\begin{figure}[t]
	\centering
	\begin{subfigure}{0.5\textwidth}
		\centering
		\includegraphics[width=\textwidth]{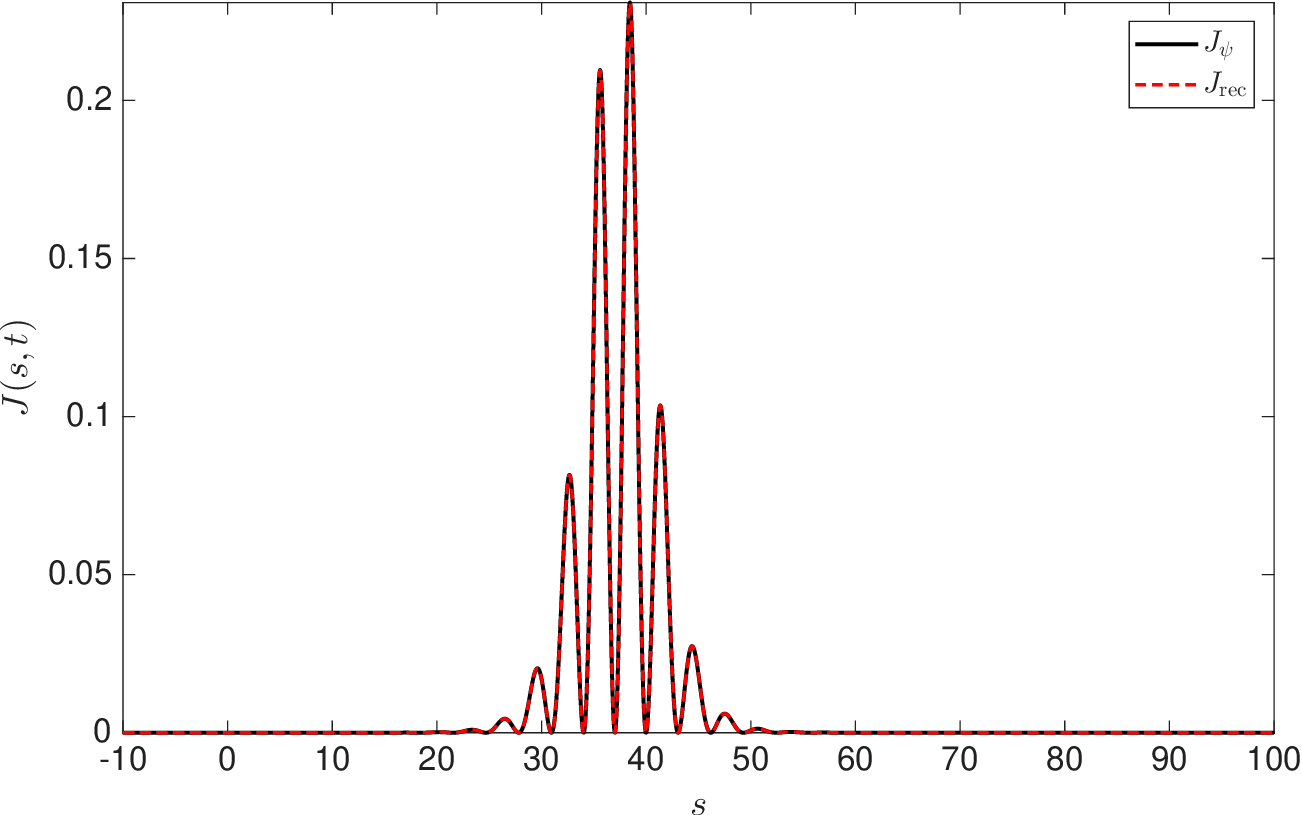}
		\caption{}		
		\label{fig:Jt20}
	\end{subfigure} 
	\hfill
	\begin{subfigure}{0.485\textwidth}
		\centering
		\includegraphics[width=\textwidth]{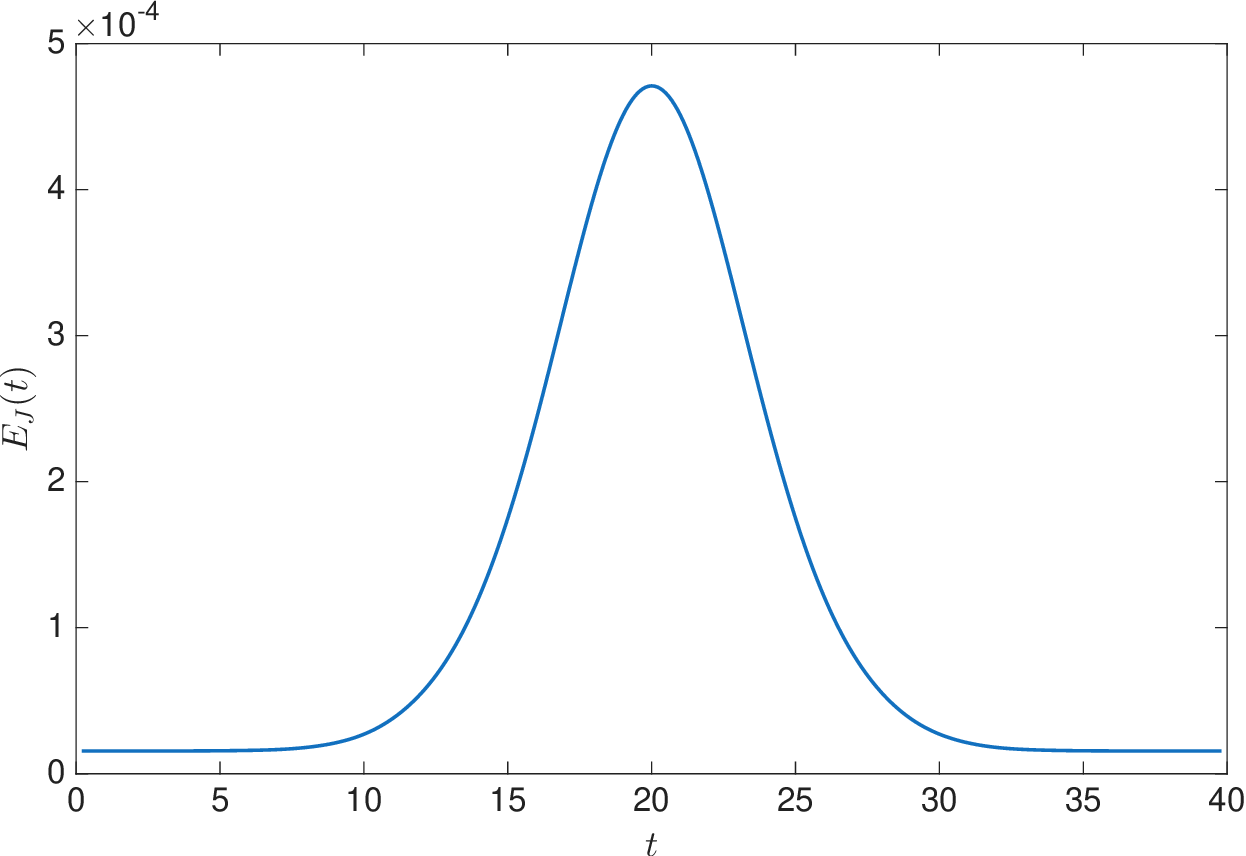}
		\caption{}	
		\label{fig:ErrJ}	
	\end{subfigure}
	\caption{Density-based reconstruction of the probability current for the
		bright two-soliton solution.
		(A) Probability current computed directly from the complex wave function,
		$J_\psi=\sigma\operatorname{Im}(\psi^*\psi_s)$, and reconstructed from the
		density evolution through \eqref{eq:Jrec}, at the interaction time $t=20$.
		(B) Relative reconstruction error $E_J(t)$ over the simulated time interval.
		The reconstruction remains accurate during the nonlinear interaction, with
		$E_J(20)\simeq4.71\times10^{-4}$.}
	\label{fig:dynamicdensity}
\end{figure}

This computation illustrates a concrete consequence of the hydrodynamic
formulation. While a single density snapshot does not determine the phase-derived current, a time-resolved density determines the associated probability current through the continuity equation, up to the prescribed boundary condition used in \eqref{eq:Jrec}. In the present example, this reconstruction agrees closely with the current obtained from the full complex Ivancevic wave
function even during nonlinear soliton interaction. Thus the hydrodynamic
formulation provides a dynamical quantity that is not accessible from a
static density comparison alone.

	\section{Discussion and outlook}
	\label{sec:Section4}
In this article, we have established a Betchov-type hydrodynamic formulation
of the Ivancevic option-pricing equation. The resulting density, current and
momentum flux satisfy a pair of conservation laws whose coefficients are
determined explicitly by the volatility and adaptive-potential parameters of
the model. This formulation provides a direct structural correspondence
between the Ivancevic adaptive-wave equation and the Hasimoto--Betchov
description of vortex-filament dynamics, while retaining a natural financial
interpretation of the hydrodynamic variables. Beyond this correspondence, the
conservation laws provide a direct means of reconstructing the probability current from
time-resolved density data.

The construction also clarifies the role of the coefficients. For the more general NLS equation \eqref{eq:ivancevic-nls-gen}, the nonlinear-pressure term appearing in momentum flux \eqref{eq:flux-gen} is determined by $D\beta$, while the dispersive term is determined by $D^2$. In the Ivancevic normalization $D=\sigma/2$, these coefficients become $\sigma\beta/2$ and $\sigma^2/4$, respectively. This coefficient dependence is useful when comparing the Ivancevic equation with the Hasimoto normalization or with defocusing variants \cite{Ding1998}. 

The formulation developed here also suggests several natural directions for future work. One immediate direction concerns forced or regulatory versions of the Ivancevic equation considered in  \cite{Delisle2026}.  If an external potential $V(s)$ is added,
\begin{equation}
	i\psi_t+\frac{\sigma}{2}\psi_{ss}
+\beta|\psi|^2\psi	+V(s)\psi=0,
\end{equation}
the first conservation law equation \eqref{eq:convlaw1} remains unchanged, whereas the second equation becomes forced,
\begin{equation}
	(\rho u)_t+
	\partial_s
	\left[
	\rho u^2
	-\frac{\sigma\beta}{2}\rho^2
	-\frac{\sigma^2}{4}\rho(\log\rho)_{ss}
	\right]
	=
	\sigma\rho V_s .
\end{equation}
For a harmonic regulatory potential, this produces a restoring force acting on the density. Such a forced Betchov system may provide a useful framework for studying regulated or constrained price-wave dynamics.

The relation with the classical Black--Scholes framework provides another
natural direction for future work. Recently, Shushi \cite{Shushi2026}
reformulated the linear Black--Scholes equation in quantum-hydrodynamic
variables via a Wick rotation and the Madelung transformation. In the present
paper, the illustrative calibration of Subsection~\ref{subsec:calibration}
compares the normalized Ivancevic one-soliton density with the Black--Scholes
risk-neutral density and yields an explicit relation between the Ivancevic
coefficient ratio and standard market quantities. As noted in
Subsection~\ref{subsec:calibration}, this comparison is static and involves
the density alone. By contrast, the dynamical reconstruction in
Subsection~\ref{subsec:dynamiccurrent} shows that, when the temporal evolution of the density is available, the associated probability current can be
recovered through the continuity equation, subject to a prescribed boundary
current. Thus, time-resolved density information provides access to the
probability transport represented by the hydrodynamic current without
requiring direct knowledge of the phase of the Ivancevic wave function.

A natural next step is to replace the synthetic model densities considered
here by market-implied risk-neutral densities extracted from option prices
via the Breeden--Litzenberger formula and to investigate the corresponding
density--current dynamics. A more systematic comparison with the Black--Scholes quantum-hydrodynamic formulation, particularly the roles of the probability current, momentum flux and quantum potential, is also left for future work.

Another possible extension concerns stochastic volatility. Ivancevic's two-component formulation already treats volatility and option price as coupled wave fields. In the present paper we considered the fixed-polarization case, where the total density satisfies a scalar Betchov-type system. More general vector solutions may have component-dependent phases and developing a complete hydrodynamic theory for such vectorial solutions is left for future work.
 
In summary, the present paper develops a coefficient-explicit Betchov-type
hydrodynamic formulation of the Ivancevic option-pricing equation and
identifies its correspondence with the Hasimoto--Betchov formulation of
vortex-filament dynamics. The formulation gives the density, probability
current and momentum flux a natural financial interpretation while making
explicit how the option-pricing coefficients enter the associated transport
structure. Beyond this structural correspondence, the continuity equation allows the probability current to be reconstructed from time-resolved density data. For the interacting two-soliton solution considered here, good numerical
agreement is observed between the reconstructed current and that computed
directly from the complex wave function.

From the financial perspective, the static comparison with the
Black--Scholes risk-neutral density relates the normalized one-soliton profile
to standard market quantities through an explicit coefficient relation,
whereas the dynamical reconstruction illustrates the additional transport
information available when the temporal evolution of the density is taken
into account. Together, these results illustrate how the Betchov-type
hydrodynamic formulation can be used to describe both the distribution and
transport of the Ivancevic density, while providing a bridge between
nonlinear wave theory, geometric fluid mechanics and mathematical finance.

\appendix
\section{Explicit solutions and calibration formulas}
\subsection{The scalar bright soliton and its two normalizations}
\label{app:bright}
The scalar bright soliton of \eqref{eq:ivancevic-nls} is, in general,
\begin{equation}
	\psi(s,t)=A\,\operatorname{sech}\!\big(\eta(s-vt)\big)\,e^{i(ks-\omega t)},
	\qquad
	\beta A^{2}=\sigma\eta^{2},\quad v=\sigma k,\quad
	\omega=\tfrac{\sigma}{2}\big(k^{2}-\eta^{2}\big),
	\label{eq:app-bright-general}
\end{equation}
with density
\begin{equation}
	\rho(s,t)=\frac{\sigma}{\beta}\,\eta^{2}\,
	\operatorname{sech}^{2}\!\big(\eta(s-\sigma k t)\big).
	\label{eq:app-bright-rho}
\end{equation}
Here $\sigma$ is the volatility; the soliton has inverse width $\eta$ and
propagates at speed $\sigma k$. In this article, we have considered
\begin{enumerate}[label=(\roman*)]
	\item Propagating form (Subsection~\ref{subsec:brightsolitonexample}): fixing $\eta=1$ gives
	$A^{2}=\sigma/\beta$ and $\rho=(\sigma/\beta)\operatorname{sech}^{2}(s-\sigma k t)$,
	a soliton of amplitude $\sigma/\beta$ travelling at speed $\sigma k$.
	\item Calibration form (Subsection~\ref{subsec:calibration}): a single-maturity
	risk-neutral density has no propagation, so we take the stationary representative
	$k=0$ ($v=0$) and keep the width $\eta$ free.
\end{enumerate}
Since
$\int_{\mathbb{R}}\rho\,ds=2(\sigma/\beta)\eta$, the unit-mass normalization gives
\begin{equation}
	\widehat{\rho}(s)=\frac{\eta}{2}\,\operatorname{sech}^{2}\!\big(\eta(s-s_{0})\big),
	\qquad \frac{\sigma}{\beta}=\frac{1}{2\eta},
	\label{eq:app-bright-normalized}
\end{equation}
with $s_{0}$ the peak location. Thus \eqref{eq:app-bright-normalized} is the same
soliton as \eqref{eq:app-bright-rho} regarded as a probability density at fixed
time, with its width free to match a target. The propagation speed $\sigma k$ changes only the position of the soliton as time evolves; the profile at a fixed time is determined only by the width parameter $\eta.$

\subsection{Explicit two-soliton solution and parameters for Figures~\ref{fig:hirota-densities} and~\ref{fig:dynamicdensity}}
\label{app:twosoliton}

The bright $N$-soliton solutions considered in Subsection~\ref{subsec:Nbrightsolitonexample} are obtained by
Hirota's bilinear method \cite{HirotaNLS1973}, one of the standard techniques
for constructing closed-form solutions of integrable nonlinear wave equations.
Rather than solving the nonlinear Schrödinger equation directly, one writes
the solution in the form
\[
\psi_N=\frac{f_N}{g_N}\,\psi_0,
\]
where $f_N$ and $g_N$ are the associated Hirota $\tau$-functions and
$\psi_0$ is the background plane wave. The nonlinear equation is thereby
converted into a bilinear system for $f_N$ and $g_N$, whose solutions can be
constructed algebraically. Bilinear methods, together with related
auto-Bäcklund and symbolic-computation techniques, have been successfully
applied to a broad range of integrable nonlinear wave equations, including
nonlocal and complex modified Korteweg--de Vries equations and related models
\cite{Gao2026, Gao2026similarity, Gao2026+}.

The two-soliton field in Figure~\ref{fig:hirota-densities} is the analytic bright
two-soliton solution of \eqref{eq:ivancevic-nls}, given by \cite{HirotaNLS1973,Delisle2026}
$$
\psi_2=(f_2/g_2)\,e^{i(\kappa s+\omega t+\phi)}, \quad \omega=\beta-\tfrac{\sigma}{2}\kappa^{2}
$$
with
\[
f_2=e^{\theta_1}+e^{\theta_2}
+\frac{\sigma}{\beta}A_{12}\Big(a_{11}a_{21}\,e^{\theta_1+\theta_2+\theta_1^*}
+a_{12}a_{22}\,e^{\theta_1+\theta_2+\theta_2^*}\Big),
\]
\[
g_2=1+\sum_{m,n=1}^{2}a_{mn}\,e^{\theta_m+\theta_n^*}
+\Big(\tfrac{\sigma}{\beta}\Big)^{2}a_{11}a_{22}\,|a_{12}|^{2}\,|A_{12}|^{2}\,
e^{\theta_1+\theta_2+\theta_1^*+\theta_2^*},
\]
where 
\[
\Omega_j=-i\beta-\kappa\sigma\lambda_j+\tfrac{i\sigma}{2}\lambda_j^{2},\qquad
\theta_j=\lambda_j s+\Omega_j t+\Phi_j,\qquad j=1,2,
\]
\[
a_{mn}=\frac{\beta}{\sigma\,(\lambda_m+\lambda_n^*)^{2}},\qquad
A_{12}=(\lambda_1-\lambda_2)^{2}.
\]

The plotted fields are $\rho_2=|\psi_2|^{2}$ and
$J_2=\sigma\operatorname{Im}(\psi_2^{*}(\psi_2)_s)$, with parameters
$\sigma=\beta=\kappa=1$, $\phi=0$ (so that $\omega=\tfrac12$),
$\lambda_1=\tfrac14-i$, $\lambda_2=\tfrac14+i$, $\Phi_1=-10$, and
$\Phi_2=0$, for $t\in[0,40]$. The solution is computed on
$s\in[-120,140]$, while the displayed spatial window is
$s\in[-10,100]$. The spatial derivatives of $\psi_2$ appearing in the
computation of $J_2$ are evaluated using fourth-order centered finite
differences. These parameters produce two bright branches of width
$\approx 1/\operatorname{Re}\lambda=4$, one nearly stationary near
$s=40$ and one propagating with speed~2, interacting near $t=20$. 

\section*{Declaration of generative AI and AI-assisted technologies in the manuscript preparation process}
During the preparation of this work, the author used Claude (Anthropic) and ChatGPT (OpenAI) to assist with language editing, LaTeX formatting, and the debugging of the MATLAB scripts. The author reviewed and edited the output as needed and takes full responsibility for the content of the published article.

\section*{Acknowledgments}
The author thanks the anonymous reviewers for their constructive feedback, which helped substantially improve the manuscript.

\bibliography{references}
\bibliographystyle{ieeetr}
\end{document}